\documentclass[reprint, amsmath,amssymb, aps,pra,]{revtex4-1}

\usepackage{graphicx}
\usepackage{dcolumn}
\usepackage{bm}
\usepackage{hyperref}
\usepackage{tikz}
\usepackage{datetime2}

\newcommand\scalemath[2]{\scalebox{#1}{\mbox{\ensuremath{\displaystyle #2}}}}

\newtheorem{theorem}{Theorem}[section]
\newtheorem{lemma}[theorem]{Lemma}

\newenvironment{proof}[1][Proof]{\begin{trivlist}
\item[\hskip \labelsep {\bfseries #1}]}{\end{trivlist}}

\newcommand{\qed}{\nobreak \ifvmode \relax \else
      \ifdim\lastskip<1.5em \hskip-\lastskip
      \hskip0.5em plus0em minus0.5em \fi \nobreak
      \vrule height0.75em width0.5em depth0.25em\fi}

\begin{document}

\title{Non-Locality distillation in tripartite NLBs}
\thanks{Non-local boxes}%

\author{Talha Lateef}
\email{talha786@msn.com}

\date{2017-11-1}

\begin{abstract}
%An article usually includes an abstract, a concise summary of the work
%covered at length in the main body of the article. 
In quantum mechanics some spatially separated sub-systems behave as if they are part of a single system, the superposition of states of this single
system cannot be written as products of states of individual sub-systems,we say that the state of such system is \texttt{entangled}, such systems give rise to
non-local correlations between outcomes of measurements. The non-local correlations are conditional probability distributions of some measurement outcomes given
some measurement settings and cannot be explained by shared information\cite{Bell:111654,PhysRevLett.23.880}.These correlations can be studied using a \texttt{non-local box(NLB)} which can be viewed as a quantum system\footnote{we assume quantum correlations.}. 
A NLB is an abstract object which has number of inputs(measurement settings) and number of outputs(outcomes), such NLBs can be both quantum and super-quantum\cite{Cs,Popescu1994,PhysRevA.71.022101}.
The correlations are of use in quantum information theory, the stronger the correlations the more useful they are, hence we study protocols that have multiple weaker non-local
systems, application of these protocols to weaker systems may result in stronger non-local correlations, we call such protocols \texttt{non-locality distillation protocols}\cite{PhysRevLett.102.120401,PhysRevLett.102.160403,PhysRevA.80.062107}. In our work here we present non-locality distillation protocols for \texttt{\hyperref[INTRO]{tripartite NLBs}} specifically \texttt{\hyperref[GHZ]{GHZ box}} and \texttt{\hyperref[START]{class 44,45 and 46}}  of no-signalling polytope.
\end{abstract}

\pacs{Valid PACS appear here}% PACS, the Physics and Astronomy
\maketitle

\section{\label{INTRO}Introduction:\protect}

The no-signaling box under consideration has three parties, two measurement settings, two measurement outcomes and is represented by the tuple $(3,2,2)$. The measurement settings/inputs are $x,y,z \in \{0,1\}$ and for measurement outcomes/outputs we use either $a,b,c \in \{0,1\}$ or the so called fourier representation $a,b,c \in \{1,-1\}$. The three parties are Alice with input $x$ and output $a$, Bob with input $y$ and output $b$ and Charlie with input $z$ and output $c$. The correlations are in the form of conditional probabilities $P(abc|xyz)$, these correlations satisfy three constraints of \texttt{positivity} , \texttt{normalization} and \texttt{no-signaling} respectively:

\begin{flalign*}   
&P(abc|xyz) \geq 0 \text{\ \ \ \ for all $a,b,c,x,y,z$}\\
&\sum_{a,b,c} P(abc|xyz)=1 \text{\ \ \ \ for all $x,y,z$}\\\\
&\text{and} \\\\
&\sum_{a}P(abc|xyz)=\sum_{a}P(abc|x^{\prime}yz) \text{\ \ \ \ for all $b,c,x,x^{\prime},y,z$}\\
&\sum_{b}P(abc|xyz)=\sum_{b}P(abc|xy^{\prime}z) \text{\ \ \ \ for all $a,c,x,y,y^{\prime},z$}\\
&\sum_{c}P(abc|xyz)=\sum_{c}P(abc|xyz^{\prime}) \text{\ \ \ \ for all $a,b,x,y,z,z^{\prime}$}\\
\end{flalign*}

Positivity and normalization comes from the correlations being conditional probabilities. No-signaling prevents faster than light communication between parties. For the case of non-local tripartite boxes they also can not be written in terms of shared information and correlations of individual parties so they do not satisfy:

\begin{equation*}
P(abc|xyz)=\sum_{\lambda}p(\lambda)P(a|x,\lambda)P(b|y,\lambda)P(c|z,\lambda)
\end{equation*}
where $p(\lambda)$ is probability distribution or shared information, we are using a probability distribution because the variable(s) may be beyond our control, from this it follows that  $p(\lambda) \geq 0$ and $\sum_{\lambda}p(\lambda) = 1$. The constraints of positivity, normalization and no-signaling is a description of a geometric object called \texttt{polytope},   it is possible to change one description of polytope to another, application of vertex enumeration algorithms to the constraints will result in all vertices or boxes. There are total 53856 boxes with 64 being local and rest of them non-local.These boxes are classified in 46 equivalent classes\cite{1751-8121-44-6-065303,PhysRevA.88.014102}.

\texttt{Non-locality distillation protocol} combines multiple non-local boxes to produce stronger non-locality using local operations without any classical communication. The \texttt{depth} of the protocol is the number of non-local boxes that are part of the protocol. In the case of depth $n$ protocol the three parties Alice , Bob and Charlie input three bits $x$,$y$ and $z$ to the protocol, the input to the first box is the same as the input to the protocol i.e. $x_{1}=x$ , $y_{1}=y$ and  $z_{1}=z$, the output bits from the first box are $a_{1}$ , $b_{1}$ and  $c_{1}$. The input to $n$th box are bits $x_{n}$ , $y_{n}$ and  $z_{n}$, which are result of each party using a local operation on their input bit and output bit from the $n-1$th box i.e.  $x_{n}=f_{n}(x_{n-1},a_{n-1})$ , $y_{n}=g_{n}(y_{n-1},b_{n-1})$ and $z_{n}=h_{n}(z_{n-1},c_{n-1})$, we get the output $a_{n}$ , $b_{n}$ and  $c_{n}$ from the $n$th box. The output of the protocol is then define as $a=f_{n+1}(x,a_{1},...,a_{n})$ , $b=g_{n+1}(y,b_{1},...,b_{n})$ and $c=h_{n+1}(z,c_{1},...,c_{n})$.

A protocol is said to be \texttt{non-adaptive} if all the inputs to the NLBs in the protocol depends on the original input to the protocol i.e. $x_{n}=f_{n}(x)$ , $y_{n}=g_{n}(y)$ and $z_{n}=h_{n}(z)$ for all $n$ while in the case of \texttt{adaptive} protocols inputs to NLBs also depend on outputs from the previous NLBs in the protocol which can be seen in the original definition, hence non-adaptive protocols are a subset of adaptive protocols\cite{PhysRevLett.102.120401,PhysRevA.80.062107,PhysRevA.82.042118}.

\begin{figure}[t]
  \centering
  \begin{tikzpicture}
    \draw (-4.3,-4.5) rectangle (4.3,4.5);    
   \node[inner sep=0pt] (russell) at (0,0)
    {\includegraphics[width=0.45\textwidth]{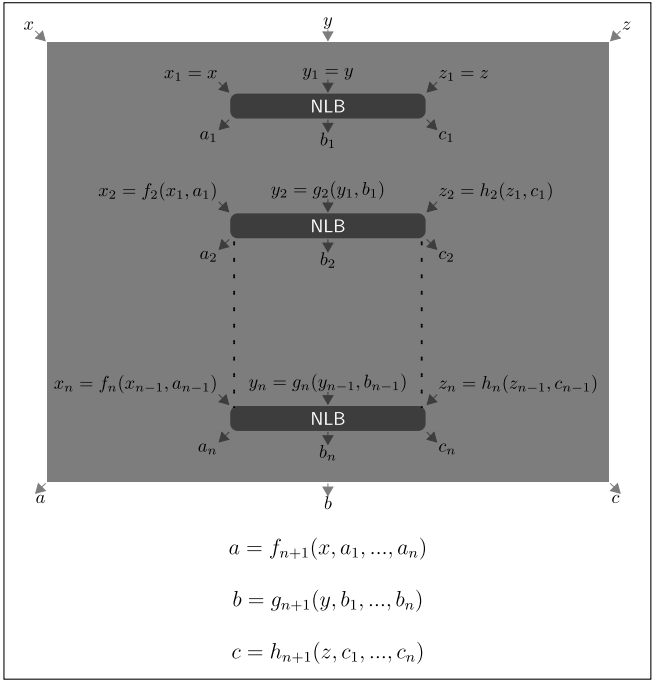}}; 
   \end{tikzpicture}  
  \caption{Distillation protocol of depth n}
  \label{fig:figure1}
\end{figure}

\section{\label{GHZ}GHZ box:\protect}
This NLB is based on GHZ correlations\cite{doi:10.1119/1.16243} and is defined as:

\begin{equation*}
P_{abc|xyz=000}=
	\begin{cases}
	\frac{1}{4} & \text{if } a \oplus b \oplus c=0\\ 
	0 & \text{if } a \oplus b \oplus c=1\\
	\end{cases}
\end{equation*}

\begin{equation*}
P_{abc|xyz\neq000}=
	\begin{cases}
	\frac{1}{4} & \text{if } a \oplus b \oplus c=1\\ 
	0 & \text{if } a \oplus b \oplus c=0\\
	\end{cases}
\end{equation*}\\

and in matrix form:
\begin{equation*}
P_{GHZ}=\left(\begin{array}{cccccccc}\frac{1}{4}&0&0&\frac{1}{4}&0&\frac{1}{4}&\frac{1}{4}&0\\0&\frac{1}{4}&\frac{1}{4}&0&\frac{1}{4}&0&0&\frac{1}{4}\\0&\frac{1}{4}&\frac{1}{4}&0&\frac{1}{4}&0&0&\frac{1}{4}\\0&\frac{1}{4}&\frac{1}{4}&0&\frac{1}{4}&0&0&\frac{1}{4}\end{array}\right)
\end{equation*}
The inputs are limited to even parity i.e $xyz\in \{000,011,101,110\}$ and outputs $a,b,c\in \{0,1\}$.

The GHZ box is not an extremal point of no-signaling polytope, however It is possible to construct GHZ boxes using boxes of class 46 as given in \cite{1751-8121-44-6-065303} by using:

\begin{equation*}
P_{GHZ}=\frac{1}{2}P_{46} + \frac{1}{2}P_{46}^{\prime}
\end{equation*}
where $P_{46}$ and $P_{46}^{\prime}$ are both different boxes of class 46. The GHZ box with even parity inputs can be constructed in similar way using boxes:

\begin{align*}
&P_{46}=\frac{1}{8}(1+abc(-1)^{xy+xz+yz})\\
\\&\text{and}\\\\
&P_{46}^{\prime}=\frac{1}{8}(1+abc(-1)^{x+y+z+xy+xz+yz})\\
&\ \ \ \ \ \ \ \ \ \text{where } a,b,c \in \{1,-1\}\\ 
\end{align*}
or just using $P_{44}$ where $P_{44}$ is the extremal point of class 44 given by:
\begin{align*}
P_{44}=\frac{1}{8}(1+abc(-1)&^{x+y+z+xy+xz+yz+xyz})\\
\text{where } a,b,c &\in \{1,-1\}\\\\ 
\end{align*}

\begin{figure}[t]
  \centering
  \begin{tikzpicture}
    \draw (-4,-2.0) rectangle (4,2.0);
        \node[inner sep=0pt] (russell) at (0,0)
    {\includegraphics[width=0.3\textwidth]{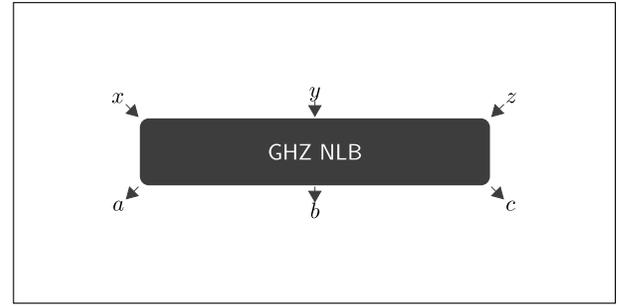}}; 
      \end{tikzpicture}  
  \caption{Noisy GHZ box}
  \label{fig:figure2}
\end{figure}

The local deterministic vertices of no-signalling polytope are grouped together as local polytope $P_{L}$ and has 64 vertices which is expressed as:

\begin{align*}
P_{L}^{\iota\kappa\mu\nu\sigma\tau}(abc|xyz)&=
\begin{cases}
1 & if\:a=\iota x\oplus\kappa,\:b=\mu y\oplus\nu,\:c=\sigma z\oplus\tau\\
0 & otherwise
\end{cases}\\
&\text{where } \iota,\kappa,\mu,\nu,\sigma,\tau \in \{0,1\}\\\\
\end{align*}

The $P_{L}$ results in 53856 inequalities which has been classified in 46 different classes\cite{PhysRevA.64.014102,quant-ph/0305190}, the amount of non-locality or box value for GHZ box is obtained using the bell inequalities of class 2 :
\begin{equation}
\label{ghzineq}
\scalemath{1.0}{
\normalsize{-2 \leq \left\langle A_{0} B_{0} C_{0} \right\rangle-\left\langle A_{0} B_{1} C_{1} \right\rangle-\left\langle A_{1} B_{0} C_{1} \right\rangle-\left\langle A_{1} B_{1} C_{0} \right\rangle \leq 2}}
%\normalsize
\end{equation}

where
\begin{equation*}
\left\langle A_{x} B_{y} C_{z} \right\rangle=\sum_{a,b,c \in \{0,1\}}(-1)^{a \oplus b \oplus c}P(abc|xyz)
\end{equation*} 

\section{GHZ distillation:\protect}
We now present the noisy form of a GHZ box with parameters $\epsilon \in [1,-1]$ and $\delta \in [1,-1]$

\begin{equation*}
P_{abc|xyz=000}=
	\begin{cases}
	\frac{1+\epsilon}{8} & \text{if } a \oplus b \oplus c=0\\ 
	\frac{1-\epsilon}{8} & \text{if } a \oplus b \oplus c=1\\
	\end{cases}
\end{equation*}

\begin{equation*}
P_{abc|xyz\neq000}=
	\begin{cases}
	\frac{1+\delta}{8} & \text{$if$ } a \oplus b \oplus c=0\\ 
	\frac{1-\delta}{8} & \text{$if$ } a \oplus b \oplus c=1\\
	\end{cases}
\end{equation*}

\begin{figure}[t]
  \centering
  \begin{tikzpicture}
    \draw (-4.2,-4.5) rectangle (4.2,4.5);
    %\node (naveq2) [nlb] at (0, 0) {\textsf{GHZ NLB} \\ $a_2 \oplus b_2 \oplus b_2 = x_2 \vee y_2 \vee z_2$};
    \node[inner sep=0pt] (russell) at (0,0)
    {\includegraphics[width=0.4\textwidth]{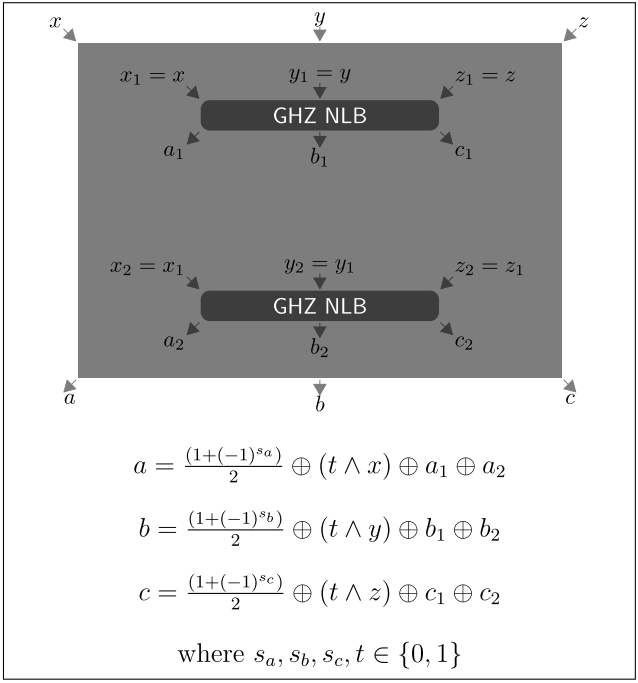}};     
  \end{tikzpicture}  
  \caption{General depth 2 distillation protocol}
  \label{fig:figure3}
\end{figure}
and in matrix form:
\begin{equation*}
%\tiny
{ \scalemath{0.8}{P_{GHZ(\epsilon,\delta)}}=\scalemath{0.8}{\frac{1}{8}}
\left(\scalemath{0.8}{\begin{array}{cccccccc} 1+\epsilon & 1-\epsilon & 1-\epsilon & 1+\epsilon & 1-\epsilon & 1+\epsilon & 1+\epsilon & 1-\epsilon\\ 
%\frac{1}{8} & \frac{1}{8} & \frac{1}{8} & \frac{1}{8} & \frac{1}{8} & \frac{1}{8} & \frac{1}{8} & \frac{1}{8}\\ 
%\frac{1}{8} & \frac{1}{8} & \frac{1}{8} & \frac{1}{8} & \frac{1}{8} & \frac{1}{8} & \frac{1}{8} & \frac{1}{8}\\ 
1+\delta & 1-\delta & 1-\delta & 1+\delta & 1-\delta & 1+\delta & 1+\delta & 1-\delta\\
1+\delta & 1-\delta & 1-\delta & 1+\delta & 1-\delta & 1+\delta & 1+\delta & 1-\delta\\
1+\delta & 1-\delta & 1-\delta & 1+\delta & 1-\delta & 1+\delta & 1+\delta & 1-\delta\\
%1-\delta & 1+\delta & 1+\delta & 1-\delta & 1+\delta & 1-\delta & 1-\delta & 1+\delta\\
%1-\delta & 1+\delta & 1+\delta & 1-\delta & 1+\delta & 1-\delta & 1-\delta & 1+\delta\\
%1-\delta & 1+\delta & 1+\delta & 1-\delta & 1+\delta & 1-\delta & 1-\delta & 1+\delta\\ 
%\frac{1}{8} & \frac{1}{8} & \frac{1}{8} & \frac{1}{8} & \frac{1}{8} & \frac{1}{8} & \frac{1}{8} & \frac{1}{8}\\ 
%0 & \frac{1}{4} & \frac{1}{4} & 0 & \frac{1}{4} & 0 & 0 & \frac{1}{4}\\ 
%0 & \frac{1}{4} & \frac{1}{4} & 0 & \frac{1}{4} & 0 & 0 & \frac{1}{4}\\ 
%\frac{1}{8} & \frac{1}{8} & \frac{1}{8} & \frac{1}{8} & \frac{1}{8} & \frac{1}{8} & \frac{1}{8} & \frac{1}{8} 
\end{array}}\right)}
\end{equation*}
using inequality~\eqref{ghzineq} we get the single box value:
\begin{equation*}
V=\epsilon-3\delta
\end{equation*}
To be able to achieve distillation we want to attain $V^{\prime}$ such that:
\begin{equation*}
V^{\prime}>V
\end{equation*}

\begin{lemma}
In a depth n GHZ setting there is no adaptive distillation protocol that is more optimal than the non-adaptive distillation protocol.
\end{lemma}

\begin{proof}
Consider a general depth $n$ distillation protocol(see Fig.~\ref{fig:figure1}) with GHZ NLBs, the inputs to the protocol $x$,$y$ and $z$ are restricted to $\{000,011,101,110\}$ by the definition of distillation protocol and GHZ box, now suppose that the protocol is non-adaptive i.e. all of the $x_k=f_k(x_{k-1},a_{k-1}) = x_{k-1}$ , $y_k=g_k(y_{k-1},b_{k-1}) = y_{k-1}$ and $z_k=h_k(z_{k-1},c_{k-1}) = z_{k-1}$ (where $k \leq n$) are true, in this case each box gets the same input as that of the protocol and hence every input contribute towards $V^{\prime}$, we now suppose an adaptive protocol i.e. atleast one of the $x_k=f_k(x_{k-1},a_{k-1}) \ne x_{k-1}$ , $y_k=g_k(y_{k-1},b_{k-1}) \ne y_{k-1}$ and $z_k=h_k(z_{k-1},c_{k-1}) \ne z_{k-1}$ is true, which results in the input being changed to one of the $\{001,010,100,111\}$ and therefore will not be accepted by the NLB and hence will not contribute towards $V^{\prime}$.\begin{flushright}$\qed$\end{flushright}
\end{proof}

We now consider depth 2 non-adaptive protocols, the only protocols that distill are parity protocols i.e the final output functions $f_{3}$,$g_{3}$ and $h_{3}$ are polynomials of degree at most $1$ over $GF(2)$/$\mathbb{Z}_{2}$($\lbrace 0,1 \rbrace\,\oplus,\wedge)$) of outputs and or input(protocol).

Since all distillation protocols are parity protocols we can generalize protocols in a single protocol diagram (see Fig.~\ref{fig:figure3}) by introducing further boolean variables $s_{a}$,$s_{b}$,$s_{c}$ and $t$.
The value attain by this protocol is then:
\begin{equation*}
V^{\prime}=(-1)^{1 \oplus s_{a} \oplus s_{b} \oplus s_{c}} \epsilon^2+(-1)^{s_{a} \oplus s_{b} \oplus s_{c}} 3\delta^2
\end{equation*}

The distillation plot is given in Fig.~\ref{fig:figure4} the light gray curved part in the plot highlights the region of distillation.

We now turn our attention to a specific non-adaptive depth $n$ distillation protocol $NDP_n(GHZ)$(see Fig.~\ref{fig:figure5}) this is the tripartite version of Forster et al.'s\cite{PhysRevLett.102.120401} parity protocol with GHZ boxes instead of PR boxes, just as we can represent an NLB in the matrix form, we can consider the protocol as a new NLB with it's inputs($x$,$y$,$z$) and outputs($a$,$b$,$c$) without worrying about it's internal working, this neat abstraction allows us to write matrix form of the protocol as an NLB.

\begin{figure}[t]
  \centering
  \begin{tikzpicture}
    \draw (-4.2,-3.0) rectangle (4.2,3.0);
     \node[inner sep=0pt] (russell) at (0,0)
    {\includegraphics[width=0.45\textwidth]{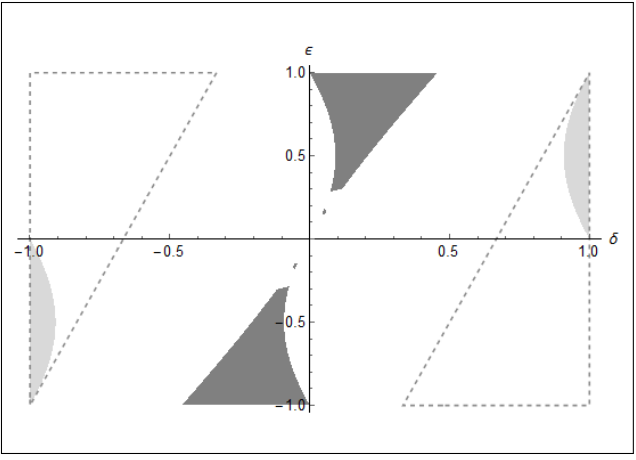}}; 
     \end{tikzpicture}  
  \caption{Distillation plot for depth 2 protocol with $V^{\prime}=\epsilon^{2}-3\delta^{2}$}
  \label{fig:figure4}
\end{figure}

\begin{lemma}
For the distillation protocol $NDP_n(GHZ)$ the matrix form is:
\begin{equation*}
\scalemath{0.75} {
%P_{NDP_n(GHZ)(\delta,\epsilon)}=
%			\begin{cases}
%			\frac{1}{8}
%			\left({\begin{array}{cccccccc} 1+\epsilon^n & 1-\epsilon^n & 1-\epsilon^n & 1+\epsilon^n & 1-\epsilon^n & 1+\epsilon^n & 1+\epsilon^n & 1-\epsilon^n\\ 
%			1-\delta^n & 1+\delta^n & 1+\delta^n & 1-\delta^n & 1+\delta^n & 1-\delta^n & 1-\delta^n & 1+\delta^n\\
%			1-\delta^n & 1+\delta^n & 1+\delta^n & 1-\delta^n & 1+\delta^n & 1-\delta^n & 1-\delta^n & 1+\delta^n\\
%			1-\delta^n & 1+\delta^n & 1+\delta^n & 1-\delta^n & 1+\delta^n & 1-\delta^n & 1-\delta^n & 1+\delta^n\\
%			\end{array}}\right) \text{$if$ $n$ odd} \\
%			\\
%			\frac{1}{8}
%			\left({\begin{array}{cccccccc} 1+\epsilon^n & 1-\epsilon^n & 1-\epsilon^n & 1+\epsilon^n & 1-\epsilon^n & 1+\epsilon^n & 1+\epsilon^n & 1-\epsilon^n\\ 
%			1+\delta^n & 1-\delta^n & 1-\delta^n & 1+\delta^n & 1-\delta^n & 1+\delta^n & 1+\delta^n & 1-\delta^n\\
%			1+\delta^n & 1-\delta^n & 1-\delta^n & 1+\delta^n & 1-\delta^n & 1+\delta^n & 1+\delta^n & 1-\delta^n\\
%			1+\delta^n & 1-\delta^n & 1-\delta^n & 1+\delta^n & 1-\delta^n & 1+\delta^n & 1+\delta^n & 1-\delta^n\\
%			\end{array}}\right) \text{$if$ $n$ even} \\			
%			\end{cases}
P_{NDP_n(GHZ)(\epsilon,\delta)}=
			\frac{1}{8}
			\left({\begin{array}{cccccccc} 1+\epsilon^n & 1-\epsilon^n & 1-\epsilon^n & 1+\epsilon^n & 1-\epsilon^n & 1+\epsilon^n & 1+\epsilon^n & 1-\epsilon^n\\ 
			1+\delta^n & 1-\delta^n & 1-\delta^n & 1+\delta^n & 1-\delta^n & 1+\delta^n & 1+\delta^n & 1-\delta^n\\
			1+\delta^n & 1-\delta^n & 1-\delta^n & 1+\delta^n & 1-\delta^n & 1+\delta^n & 1+\delta^n & 1-\delta^n\\
			1+\delta^n & 1-\delta^n & 1-\delta^n & 1+\delta^n & 1-\delta^n & 1+\delta^n & 1+\delta^n & 1-\delta^n\\
			\end{array}}\right) \\	
}
\end{equation*}
with box value: 
\begin{equation*}
V^{\prime}_{NDP_n(GHZ)}=
\epsilon^{n}-3\delta^{n}
\end{equation*}
Let $V^{\prime}_n$ be the value for any depth $n$ non-adaptive GHZ protocol, then $V^{\prime}_n \leq V^{\prime}_{NDP_n(GHZ)}$ i.e. $V^{\prime}_n$ is bounded by $\pmb{max_{1 \leq k \leq n }\left|\epsilon^{k}-3\delta^{k}\right|}$
\end{lemma}

\begin{proof}
We use Fourier transform for boolean function in our proof\cite{PhysRevA.82.042118}.
Let the $n$ bit tuples $(a_{i},b_{i},c_{i})$ be the output of the $n$ NLBs that Alice, Bob and Charlie obtain for inputs $x$,$y$ and $z$ respectively. The tuple $(a_{i},b_{i},c_{i})$ is drawn from $\{000,001,010,011,100,101,110,111\}$ with respect to the distribution $ \mu = \frac{1}{8}\{1+\epsilon,1-\epsilon,1-\epsilon,1+\epsilon,1-\epsilon,1+\epsilon,1+\epsilon,1-\epsilon\}$, where $\epsilon$ is the bias for the row corresponding to the inputs received by the players($x=0,y=0,z=0$ in this case). For inputs $x$,$y$ and $z$ let $\bar{A}$,$\bar{B}$,$\bar{C } \subseteq \{0,1\}^n$ be the set of strings for which Alice, Bob and Charlie's final output is $0$.

\begin{widetext}
Given that Alice, Bob and Charlie input bits $x$,$y$ and $z$ into $n$ NLBs, the probability that they receive bit strings $a$,$b$ and $c$ of length $n$ is given by:
\begin{flalign*}  
\scalemath{0.85}{P_{abc|xyz}} & \scalemath{0.85}{= \prod\limits _{1\leq i\leq n}\left(\frac{1-\epsilon}{8}+\frac{\epsilon}{4}\bigl[\thinspace\overline{a_{_{i}}\oplus b_{_{i}}\oplus c_{i}}=1\bigr]\right)}\\
\scalemath{0.85}{P_{abc|xyz}} & \scalemath{0.85}{= \left(\frac{1-\epsilon}{8}\right)^{n}\prod\limits _{1\leq i\leq n}\left(1+\frac{2\epsilon}{(1-\epsilon)}\bigl[\thinspace\overline{a_{_{i}}\oplus b_{_{i}}\oplus c_{i}}=1\bigr]\right)}\\
\scalemath{0.85}{P_{abc|xyz}} & \scalemath{0.85}{= \left(\frac{1-\epsilon}{8}\right)^{n}\left(\frac{1+\epsilon}{1-\epsilon}\right)^{n-\bigl|a\oplus b\oplus c\bigr|}}\\
\scalemath{0.85}{P_{abc|xyz}} & \scalemath{0.85}{= \frac{1}{8^{n}}\left(1-\epsilon\right)^{\bigl|a\oplus b\oplus c\bigr|}\left(1+\epsilon\right)^{n-\bigl|a\oplus b\oplus c\bigr|}}\\
\end{flalign*}
The probability $q_{(\bar{A},\bar{B},\bar{C})}(\epsilon)$ of obtaining output $000$ is obtained by summing over all of the output bit strings $a$ in $\bar{A}$, $b$ in $\bar{B}$ and $c$ in $\bar{C}$:
\begin{flalign*}
\scalemath{0.85}{q_{(\bar{A},\bar{B},\bar{C})}(\epsilon)} & \scalemath{0.85}{= \frac{1}{8^{n}}\sum_{a\in\bar{A}}\sum_{b\in\bar{B}}\sum_{c\in\bar{C}}\left(1-\epsilon\right)^{\bigl|a\oplus b\oplus c\bigr|}\left(1+\epsilon\right)^{n-\bigl|a\oplus b\oplus c\bigr|}}\\
\intertext{Representing the above expression in terms of $2^{n}$ characters of finite group ($\mathbb{Z}_{2}^{n}$)}
\scalemath{0.85}{q_{(\bar{A},\bar{B},\bar{C})}(\epsilon)} & \scalemath{0.85}{= \frac{1}{8^{n}}\sum_{a\in\bar{A}}\sum_{b\in\bar{B}}\sum_{c\in\bar{C}}\sum_{z\in\left\{ 0,1\right\} ^{n}}\chi_{z}\left(a\oplus b\oplus c\right)\epsilon^{\bigl|z\bigr|}}\\
\end{flalign*}
Using the homomorphism from additive group to multiplicative group
\begin{flalign*}
\scalemath{0.85}{q_{(\bar{A},\bar{B},\bar{C})}(\epsilon)} & \scalemath{0.85}{= \frac{1}{8^{n}}\sum_{a\in\bar{A}}\sum_{b\in\bar{B}}\sum_{c\in\bar{C}}\sum_{z\in\left\{ 0,1\right\} ^{n}}\chi_{z}(a)\chi_{z}(b)\chi_{z}(c)\epsilon^{\bigl|z\bigr|}}\\
\scalemath{0.85}{q_{(\bar{A},\bar{B},\bar{C})}(\epsilon)} & \scalemath{0.85}{= \frac{1}{8^{n}}\sum_{z\in\left\{ 0,1\right\} ^{n}}\epsilon^{\bigl|z\bigr|}\left(\sum_{a\in\bar{A}}\chi_{z}\left(a\right)\sum_{b\in\bar{B}}\chi_{z}\left(b\right)\sum_{c\in\bar{C}}\chi_{z}\left(c\right)\right)}\\
\scalemath{0.85}{q_{(\bar{A},\bar{B},\bar{C})}(\epsilon)} & \scalemath{0.85}{= \sum_{z\in\left\{ 0,1\right\} ^{n}}\epsilon^{\bigl|z\bigr|}\left(\left(\sum_{a\in\bar{A}}\frac{1}{2^{n}}\chi_{z}\left(a\right)\right)\left(\sum_{b\in\bar{B}}\frac{1}{2^{n}}\chi_{z}\left(b\right)\right)\left(\sum_{c\in\bar{C}}\frac{1}{2^{n}}\chi_{z}\left(c\right)\right)\right)}\\
\scalemath{0.85}{q_{(\bar{A},\bar{B},\bar{C})}(\epsilon)} & \scalemath{0.85}{= \sum_{z\in\left\{ 0,1\right\} ^{n}}\epsilon^{\bigl|z\bigr|}\left(\left(\sum_{s}\frac{1}{2^{n}}\chi_{z}\left(s\right)\scalemath{1.0}{\bigl[s\in\bar{A}\bigr]}\right)\left(\sum_{t}\frac{1}{2^{n}}\chi_{z}\left(t\right)\scalemath{1.0}{\bigl[t\in\bar{B}\bigr]}\right)\left(\sum_{u}\frac{1}{2^{n}}\chi_{z}\left(u\right)\scalemath{1.0}{\bigl[u\in\bar{C}\bigr]}\right)\right)}
\intertext{Let the functions $f$,$g$ and $h$ be +1 when $s$,$t$ and $u$ are in $\bar{A}$,$\bar{B}$ and $\bar{C}$, respectively and -1 otherwise:}
\scalemath{0.85}{q_{(\bar{A},\bar{B},\bar{C})}(\epsilon)} & \scalemath{0.85}{= \sum_{z\in\left\{ 0,1\right\} ^{n}}\epsilon^{\bigl|z\bigr|}\left(\left(\sum_{s}\frac{1}{2^{n}}\chi_{z}\left(s\right)\left(\frac{f\left(s\right)+1}{2}\right)\right)\left(\sum_{t}\frac{1}{2^{n}}\chi_{z}\left(t\right)\left(\frac{g\left(t\right)+1}{2}\right)\right)\left(\sum_{u}\frac{1}{2^{n}}\chi_{z}\left(u\right)\left(\frac{h\left(u\right)+1}{2}\right)\right)\right)}
\end{flalign*}
\begin{flalign*}
\intertext{we know that fourier coefficient$\hat{f}_{S}=\left\langle f,\chi_{S}\right\rangle =E_{x}\left[f(x).\chi_{S}(x)\right]=\frac{1}{2^{n}}\sum_{x\in\left\{ 0,1\right\} ^{n}}f(x)\chi_{S}(x)$
and $\hat{f}_{0}=\left\langle f,1\right\rangle$}
\scalemath{0.85}{q_{(\bar{A},\bar{B},\bar{C})}(\epsilon)} & \scalemath{0.85}{= \sum_{z\in\left\{ 0,1\right\} ^{n}}\epsilon^{\bigl|z\bigr|}\left(\left(\frac{\hat{f}_{z}+\bigl[z=0\bigr]}{2}\right)\left(\frac{\hat{g}_{z}+\bigl[z=0\bigr]}{2}\right)\left(\frac{\hat{h}_{z}+\bigl[z=0\bigr]}{2}\right)\right)}\\
\scalemath{0.85}{q_{(\bar{A},\bar{B},\bar{C})}(\epsilon)} & \scalemath{0.85}{= \frac{1}{8}\sum_{z\in\left\{ 0,1\right\} ^{n}}\epsilon^{\bigl|z\bigr|}\left(\left(\hat{f}_{z}+\bigl[z=0\bigr]\right)\left(\hat{g}_{z}+\bigl[z=0\bigr]\right)\left(\hat{h}_{z}+\bigl[z=0\bigr]\right)\right)}\\
\scalemath{0.85}{q_{(\bar{A},\bar{B},\bar{C})}(\epsilon)} & \scalemath{0.85}{= \frac{1}{8}\sum_{z\in\left\{ 0,1\right\} ^{n}}\epsilon^{\bigl|z\bigr|}\left(\left(\hat{f}_{z}\hat{g}_{z}\hat{h}_{z}+\left(1+\hat{f}_{0}+\hat{g}_{0}+\hat{h}_{0}+\hat{f}_{0}\hat{g}_{0}+\hat{f}_{0}\hat{h}_{0}+\hat{g}_{0}\hat{h}_{0}\right)\bigl[z=0\bigr]\right)\right)}
\intertext{we obtain $q_{(\bar{A},B,C)}$,$q_{(A,\bar{B},C)}$ and $q_{(A,B,\bar{C})}$ similarly by flipping sign of $\hat{f}_{z}$,$\hat{g}_{z}$ and $\hat{h}_{z}$}
\scalemath{0.85}{q_{(\bar{A},B,C)}(\epsilon)} & \scalemath{0.85}{= \sum_{z\in\left\{ 0,1\right\} ^{n}}\epsilon^{\bigl|z\bigr|}\left(\left(\frac{\hat{f}_{z}+\bigl[z=0\bigr]}{2}\right)\left(\frac{-\hat{g}_{z}+\bigl[z=0\bigr]}{2}\right)\left(\frac{-\hat{h}_{z}+\bigl[z=0\bigr]}{2}\right)\right)}\\
\scalemath{0.85}{q_{(\bar{A},B,C)}(\epsilon)} & \scalemath{0.85}{= \frac{1}{8}\sum_{z\in\left\{ 0,1\right\} ^{n}}\epsilon^{\bigl|z\bigr|}\left(\left(\hat{f}_{z}+\bigl[z=0\bigr]\right)\left(-\hat{g}_{z}+\bigl[z=0\bigr]\right)\left(-\hat{h}_{z}+\bigl[z=0\bigr]\right)\right)}\\
\scalemath{0.85}{q_{(\bar{A},B,C)}(\epsilon)} & \scalemath{0.85}{= \frac{1}{8}\sum_{z\in\left\{ 0,1\right\} ^{n}}\epsilon^{\bigl|z\bigr|}\left(\left(\hat{f}_{z}\hat{g}_{z}\hat{h}_{z}+\left(1+\hat{f}_{0}-\hat{g}_{0}-\hat{h}_{0}-\hat{f}_{0}\hat{g}_{0}-\hat{f}_{0}\hat{h}_{0}+\hat{g}_{0}\hat{h}_{0}\right)\bigl[z=0\bigr]\right)\right)}\\\\
\scalemath{0.85}{q_{(A,\bar{B},C)}(\epsilon)} & \scalemath{0.85}{= \sum_{z\in\left\{ 0,1\right\} ^{n}}\epsilon^{\bigl|z\bigr|}\left(\left(\frac{-\hat{f}_{z}+\bigl[z=0\bigr]}{2}\right)\left(\frac{\hat{g}_{z}+\bigl[z=0\bigr]}{2}\right)\left(\frac{-\hat{h}_{z}+\bigl[z=0\bigr]}{2}\right)\right)}\\
\scalemath{0.85}{q_{(A,\bar{B},C)}(\epsilon)} & \scalemath{0.85}{= \frac{1}{8}\sum_{z\in\left\{ 0,1\right\} ^{n}}\epsilon^{\bigl|z\bigr|}\left(\left(-\hat{f}_{z}+\bigl[z=0\bigr]\right)\left(\hat{g}_{z}+\bigl[z=0\bigr]\right)\left(-\hat{h}_{z}+\bigl[z=0\bigr]\right)\right)}\\
\scalemath{0.85}{q_{(A,\bar{B},C)}(\epsilon)} & \scalemath{0.85}{= \frac{1}{8}\sum_{z\in\left\{ 0,1\right\} ^{n}}\epsilon^{\bigl|z\bigr|}\left(\left(\hat{f}_{z}\hat{g}_{z}\hat{h}_{z}+\left(1-\hat{f}_{0}+\hat{g}_{0}-\hat{h}_{0}-\hat{f}_{0}\hat{g}_{0}+\hat{f}_{0}\hat{h}_{0}-\hat{g}_{0}\hat{h}_{0}\right)\bigl[z=0\bigr]\right)\right)}\\\\
\scalemath{0.85}{q_{(A,B,\bar{C})}(\epsilon)} & \scalemath{0.85}{= \sum_{z\in\left\{ 0,1\right\} ^{n}}\epsilon^{\bigl|z\bigr|}\left(\left(\frac{-\hat{f}_{z}+\bigl[z=0\bigr]}{2}\right)\left(\frac{-\hat{g}_{z}+\bigl[z=0\bigr]}{2}\right)\left(\frac{\hat{h}_{z}+\bigl[z=0\bigr]}{2}\right)\right)}\\
\scalemath{0.85}{q_{(A,B,\bar{C})}(\epsilon)} & \scalemath{0.85}{= \frac{1}{8}\sum_{z\in\left\{ 0,1\right\} ^{n}}\epsilon^{\bigl|z\bigr|}\left(\left(-\hat{f}_{z}+\bigl[z=0\bigr]\right)\left(-\hat{g}_{z}+\bigl[z=0\bigr]\right)\left(\hat{h}_{z}+\bigl[z=0\bigr]\right)\right)}\\
\scalemath{0.85}{q_{(A,B,\bar{C})}(\epsilon)} & \scalemath{0.85}{= \frac{1}{8}\sum_{z\in\left\{ 0,1\right\} ^{n}}\epsilon^{\bigl|z\bigr|}\left(\left(\hat{f}_{z}\hat{g}_{z}\hat{h}_{z}+\left(1-\hat{f}_{0}-\hat{g}_{0}+\hat{h}_{0}+\hat{f}_{0}\hat{g}_{0}-\hat{f}_{0}\hat{h}_{0}-\hat{g}_{0}\hat{h}_{0}\right)\bigl[z=0\bigr]\right)\right)}\\
\end{flalign*}
\begin{flalign*}
\intertext{The probability of obtaining output when $A\oplus B\oplus C=0$ is then:}
\scalemath{0.85}{r_{(A\oplus B\oplus C=0)}(\epsilon)} & \scalemath{0.85}{= q_{(\bar{A},\bar{B},\bar{C})}(\epsilon)+q_{(\bar{A},B,C)}(\epsilon)+q_{(A,\bar{B},C)}(\epsilon)+q_{(A,B,\bar{C})}(\epsilon)}\\
& \scalemath{0.85}{= \frac{1}{8}\sum_{z\in\left\{ 0,1\right\} ^{n}}\epsilon^{\bigl|z\bigr|}\left(4\hat{f}_{z}\hat{g}_{z}\hat{h}_{z}+4\right)}\\
& \scalemath{0.85}{= \frac{1}{2}\left(1+\sum_{z\in\left\{ 0,1\right\} ^{n}}\hat{f}_{z}\hat{g}_{z}\hat{h}_{z}\epsilon^{\bigl|z\bigr|}\right)}
\end{flalign*}
\begin{flalign*}
\intertext{The expected value for the bias $\epsilon$ can now be expressed:}
\scalemath{0.85}{r_{(A\oplus B\oplus C=0)}(\epsilon)-r_{(A\oplus B\oplus C=1)}(\epsilon)} & \scalemath{0.85}{= r_{(A\oplus B\oplus C=0)}(\epsilon)-\left(1-r_{(A\oplus B\oplus C=0)}(\epsilon)\right)}\\
& \scalemath{0.85}{= 2r_{(A\oplus B\oplus C=0)}(\epsilon)-1}\\
& \scalemath{0.85}{= 2\left[\frac{1}{2}\left(1+\sum_{z\in\left\{ 0,1\right\} ^{n}}\hat{f}_{z}\hat{g}_{z}\hat{h}_{z}\epsilon^{\bigl|z\bigr|}\right)\right]-1}\\
& \scalemath{0.85}{= \sum_{z\in\left\{ 0,1\right\} ^{n}}\hat{f}_{z}\hat{g}_{z}\hat{h}_{z}\epsilon^{\bigl|z\bigr|}}
\end{flalign*}
\begin{flalign*}
\intertext{The expected value for bias $\delta$ can be calculated similarly, using expected value of both biases in~\eqref{ghzineq} we obtain:}
\scalemath{0.85}{V^{\prime}} & \scalemath{0.85}{= \left(\sum_{z\in\left\{ 0,1\right\} ^{n}}\hat{f}_{z}\hat{g}_{z}\hat{h}_{z}\epsilon^{\bigl|z\bigr|}\right)-\left(\sum_{z\in\left\{ 0,1\right\} ^{n}}\hat{f}_{z}\hat{g}_{z}\hat{h}_{z}\delta^{\bigl|z\bigr|}\right)-\left(\sum_{z\in\left\{ 0,1\right\} ^{n}}\hat{f}_{z}\hat{g}_{z}\hat{h}_{z}\delta^{\bigl|z\bigr|}\right)-\left(\sum_{z\in\left\{ 0,1\right\} ^{n}}\hat{f}_{z}\hat{g}_{z}\hat{h}_{z}\delta^{\bigl|z\bigr|}\right)}\\
\scalemath{0.85}{V^{\prime}} & \scalemath{0.85}{= \sum_{z\in\left\{ 0,1\right\} ^{n}}\hat{f}_{z}\hat{g}_{z}\hat{h}_{z}\left(\epsilon^{\bigl|z\bigr|}-3\delta^{\bigl|z\bigr|}\right)}\\
& \scalemath{0.85}{\leq \sum_{z\in\left\{ 0,1\right\} ^{n}}\left|\hat{f}_{z}\right|\left|\hat{g}_{z}\right|\left|\hat{h}_{z}\right|\left|\epsilon^{\bigl|z\bigr|}-3\delta^{\bigl|z\bigr|}\right|}\\
& \scalemath{0.85}{\leq max_{k}\left|\epsilon^{k}-3\delta^{k}\right|\sum_{z\in\left\{ 0,1\right\} ^{n}}\left|\hat{f}_{z}\right|\left|\hat{g}_{z}\right|\left|\hat{h}_{z}\right|}\\
& \scalemath{0.85}{\leq max_{k}\left|\epsilon^{k}-3\delta^{k}\right|} &&\qed
\end{flalign*}
\end{widetext}
\end{proof}
\begin{figure}[t]
  \centering
  \begin{tikzpicture}
    \draw (-4.2,-4.5) rectangle (4.2,4.5);
     \node[inner sep=0pt] (russell) at (0,0)
    {\includegraphics[width=0.45\textwidth]{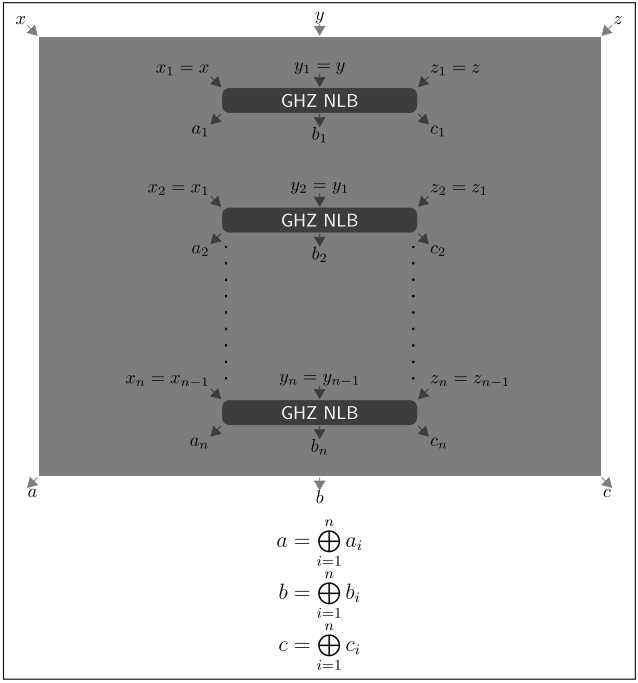}}; 
     \end{tikzpicture}
    \caption{$NDP_n(GHZ)$ distillation protocol}
  \label{fig:figure5}
\end{figure}

Using Lemma III.1 and Lemma III.2 we conclude that non-adaptive parity protocols are optimal for n copies of GHZ box.

\section{\label{START}Class 44,45 and 46:\protect}
We use the same representatives as in \cite{PhysRevA.71.022101,1751-8121-44-6-065303} and define each of class 44,45 and 46 perfect boxes as:  
\begin{flalign*}
&P^{44}=\begin{cases}
\frac{1}{4} & a\oplus b\oplus c=xyz\\
0 & otherwise
\end{cases}\\\\
&P^{45}=\begin{cases}
\frac{1}{4} & a\oplus b\oplus c=xy\oplus xz\\
0 & otherwise
\end{cases}\\\\
&P^{46}=\begin{cases}
\frac{1}{4} & a\oplus b\oplus c=xy\oplus yz\oplus xz\\
0 & otherwise
\end{cases}
\end{flalign*}
when we want to refer to any of the perfect boxes we use the notation:
\begin{flalign*}
&P^{N}=\begin{cases}
\frac{1}{4} & a\oplus b\oplus c=f(x,y,z)\\
0 & otherwise
\end{cases}
\end{flalign*}
where $N$ can be any of the perfect boxes and $f(x,y,z)$ is a boolean polynomial of $x,y,z$ over $GF(2)$ for that box.\\\\
To calculate box value we select the following inequality of class 41: 
\begin{multline}
\label{commonineq}
%\scalemath{0.8}
%{
\scalemath{0.85}{-\left\langle A_{0} \right\rangle-\left\langle B_{0} \right\rangle-\left\langle C_{0} \right\rangle+\left\langle A_{0} B_{1} \right\rangle+\left\langle A_{1} B_{0} \right\rangle -\left\langle A_{1} B_{1} \right\rangle+\left\langle A_{0} C_{1} \right\rangle}\\ 
\scalemath{0.85}{+\left\langle A_{1} C_{0} \right\rangle-\left\langle A_{1} C_{1} \right\rangle+\left\langle B_{0} C_{0} \right\rangle +3 \left\langle A_{0} B_{0} C_{0} \right\rangle+\left\langle A_{0} B_{0} C_{1} \right\rangle +\left\langle A_{0} B_{1} C_{0} \right\rangle}\\ 
\scalemath{0.85}{+2 \left\langle A_{0} B_{1} C_{1} \right\rangle+4 \left\langle A_{1} B_{0} C_{0} \right\rangle-\left\langle A_{1} B_{0} C_{1} \right\rangle-\left\langle A_{1} B_{1} C_{0} \right\rangle-2 \left\langle A_{1} B_{1} C_{1} \right\rangle\leq7}
%}
\end{multline}
calculating $V$ for each of the perfect boxes using ~\eqref{commonineq}, we get the same value of $11$.\\\\ 
$P^{c}$ is the tripartite correlated box with even parity and is defined as:
\begin{flalign*}
P^{c}=\begin{cases}
\frac{1}{4} & a\oplus b\oplus c=0\\
0 & otherwise
\end{cases}
\end{flalign*}
we are restricting ourselves to correlated noise model hence we define noisy box as:
\begin{flalign*}
P_{\delta}^{N}=&\left[\delta P^{N}+(1-\delta)P^{c}\right]\\
&\text{where $\delta \in [0,1]$}\\
\end{flalign*}
Using ~\eqref{commonineq} we get the box value $V$ for $P_{\delta}^{N}$:
\begin{equation*}
V=4\delta+7
\end{equation*}
Finally we present couple of notations that we use in next sections:
\begin{flalign*}
P_{(1-\delta)}^{N}=&\left[(1-\delta)P^{N}+\delta P^{c}\right]\\
&\text{where $\delta \in [0,1]$}
\end{flalign*}
\begin{flalign*}
&P^{f(x_{1},x_{2},.....,x_{n})}\rightarrow a\oplus b\oplus c=f(x_{1},x_{2},.....,x_{n})\\
&\text{where $f(x_{1},x_{2},.....,x_{n})$ is polynomial of inputs.}
\end{flalign*}

\section{$P_{\delta}^{N}$ distillation:\protect}
\subsection{Common protocols}
In this section we are going to present two distillation protocols of depth $2$ that are common to these three classes: 
\begin{figure}[t]
  \centering
  \begin{tikzpicture}
    \draw (-4.2,-3.0) rectangle (4.2,3.0);
     \node[inner sep=0pt] (russell) at (0,0)
    {\includegraphics[width=0.45\textwidth]{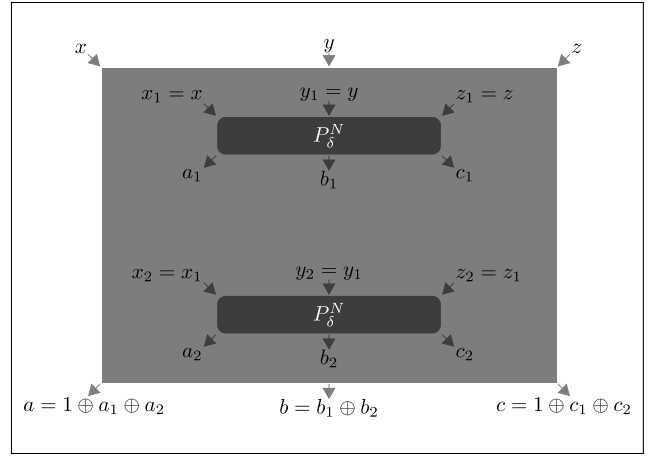}}; 
     \end{tikzpicture}
    \caption{Protocol 1}
  \label{fig:figure6}
\end{figure}
\subsubsection{Protocol 1}
This protocol(see Fig.~\ref{fig:figure6}) is the similar to the non-adaptive protocol of depth 2 which has been presented previously for GHZ box, we now give proof of it's distillability:
\begin{lemma}
Protocol 1(see Fig.~\ref{fig:figure6}) distills for two copies of $P_{\delta}^{N}$ with the same amount of $V^{\prime}$ for each class
\end{lemma}
\begin{proof}
We start with two copies of $P_{\delta}^{N}$:
\begin{flalign*}
P_{\delta}^{N}P_{\delta}^{N}&=\left[\delta P^{N}+(1-\delta)P^{c}\right]\left[\delta P^{N}+(1-\delta)P^{c}\right]\\
&=\delta^{2}P^{N}P^{N}+\delta(1-\delta)\left(P^{N}P^{c}+P^{c}P^{N}\right)+(1-\delta)^{2}P^{c}P^{c}
\end{flalign*}
Hence we need each of the four relations above to get to the final box:\\
$P^{N}P^{N}\rightarrow$ For box 1 we have $a_{1}\oplus b_{1}\oplus c_{1}=f(x_{1},y_{1},z_{1})=f(x,y,z)$, for box 2 we have $a_{2}\oplus b_{2}\oplus c_{2}=f(x_{2},y_{2},z_{2})=f(x_{1},y_{1},z_{1})=f(x,y,z)$, equating we get: $1\oplus a_{1}\oplus a_{2}\oplus b_{1}\oplus b_{2}\oplus1\oplus c_{1}\oplus c_{2}=0$, $a\oplus b\oplus c=0$, hence $P^{c}$\\\\
$P^{N}P^{c}\rightarrow$ For box 1 we have $a_{1}\oplus b_{1}\oplus c_{1}=f(x_{1},y_{1},z_{1})=f(x,y,z)$, $0=f(x,y,z)\oplus a_{1}\oplus b_{1}\oplus c_{1}$, for box 2 we have $a_{2}\oplus b_{2}\oplus c_{2}=0$, equating we get: $1\oplus a_{1}\oplus a_{2}\oplus b_{1}\oplus b_{2}\oplus1\oplus c_{1}\oplus c_{2}=f(x,y,z)$, $a\oplus b\oplus c=f(x,y,z)$, hence $P^{N}$\\\\
$P^{c}P^{N}\rightarrow$ For box 1 we have $a_{1}\oplus b_{1}\oplus c_{1}=0$, for box 2 we have $a_{2}\oplus b_{2}\oplus c_{2}=f(x_{2},y_{2},z_{2})=f(x_{1},y_{1},z_{1})=f(x,y,z)$, $0=f(x,y,z)\oplus a_{2}\oplus b_{2}\oplus c_{2}$ equating we get: $1\oplus a_{1}\oplus a_{2}\oplus b_{1}\oplus b_{2}\oplus1\oplus c_{1}\oplus c_{2}=f(x,y,z)$, $a\oplus b\oplus c=f(x,y,z)$, hence $P^{N}$\\\\
$P^{c}P^{c}\rightarrow$ For box 1 we have $a_{1}\oplus b_{1}\oplus c_{1}=0$, for box 2 we have $a_{2}\oplus b_{2}\oplus c_{2}=0$, equating we get:
$1\oplus a_{1}\oplus a_{2}\oplus b_{1}\oplus b_{2}\oplus1\oplus c_{1}\oplus c_{2}=0$, $a\oplus b\oplus c=0$, hence $P^{c}$\\\\
The final box is given by:
\begin{flalign*}
P_{\delta}^{N}P_{\delta}^{N}&=\left(\delta(2-2\delta)\right)P^{N}+\left(1-\delta\left(2-2\delta\right)\right)P^{c}
\end{flalign*}
Applying \eqref{commonineq} to the final box we get:
\begin{flalign*}
&&V^{\prime}=-8\delta^{2}+8\delta+7 &&\qed
\end{flalign*}
\end{proof}
The protocol distills($V^{\prime}>V$) in the region of $0<\delta<\dfrac{1}{2}$, the grey curve in fig.~\ref{fig:figure11} highlights the protocol.

\subsubsection{Protocol 2}
\begin{figure}[t]
  \centering
  \begin{tikzpicture}
    \draw (-4.2,-3.0) rectangle (4.2,3.0);
       \node[inner sep=0pt] (russell) at (0,0)
    {\includegraphics[width=0.45\textwidth]{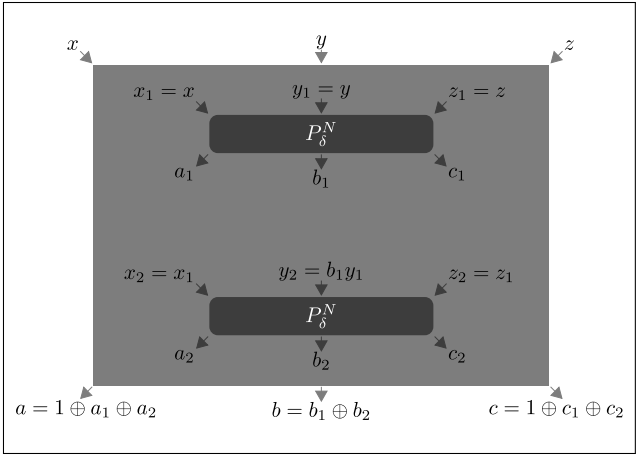}}; 
     \end{tikzpicture}
    \caption{Protocol 2}
  \label{fig:figure7}
\end{figure}
We now turn our attention to another common protocol(see fig.~\ref{fig:figure7}), protocol 2 is a depth 2 adaptive protocol with final box and expression given below: 
\begin{lemma}
Protocol 2(see fig.~\ref{fig:figure7}) distills for two copies of $P_{\delta}^{N}$ with the different amount of $V^{\prime}$ for each class
\end{lemma}
\begin{proof}
For each class we derive a different final box and expression:\\
\texttt{Class 44:}
\begin{flalign*}
\scalemath{0.9}{P_{\delta}^{44}P_{\delta}^{44}}&\scalemath{0.9}{=\delta^{2}P^{44}P^{44}+\delta(1-\delta)\left(P^{44}P^{c}+P^{c}P^{44}\right)+(1-\delta)^{2}P^{c}P^{c}}
\end{flalign*}
$P^{44}P^{44}\rightarrow$ For box 1 we have $a_{1}\oplus b_{1}\oplus c_{1}=x_{1}y_{1}z_{1}=xyz$, $0=xyz\oplus a_{1}\oplus b_{1}\oplus c_{1}$, for box 2 we have $a_{2}\oplus b_{2}\oplus c_{2}=x_{2}y_{2}z_{2}=x_{1}y_{1}z_{1}b_{1}=xyzb_{1}$, $0=xyzb_{1}\oplus a_{2}\oplus b_{2}\oplus c_{2}$ equating we get: $1\oplus a_{1}\oplus a_{2}\oplus b_{1}\oplus b_{2}\oplus1\oplus c_{1}\oplus c_{2}=xyz\oplus xyzb_{1}$, $a\oplus b\oplus c=xyz\left(1\oplus b_{1}\right)$, hence $\frac{1}{2}\left(P^{44}+P^{c}\right)$\\\\
$P^{44}P^{c}\rightarrow$ For box 1 we have $a_{1}\oplus b_{1}\oplus c_{1}=x_{1}y_{1}z_{1}=xyz$, $0=xyz\oplus a_{1}\oplus b_{1}\oplus c_{1}$, for box 2 we have $a_{2}\oplus b_{2}\oplus c_{2}=0$, equating we get: $1\oplus a_{1}\oplus a_{2}\oplus b_{1}\oplus b_{2}\oplus1\oplus c_{1}\oplus c_{2}=xyz$, $a\oplus b\oplus c=xyz$, hence $P^{44}$\\\\
$P^{c}P^{44}\rightarrow$ For box 1 we have $a_{1}\oplus b_{1}\oplus c_{1}=0$, for box 2 we have $a_{2}\oplus b_{2}\oplus c_{2}=x_{2}y_{2}z_{2}=x_{1}y_{1}z_{1}b_{1}=xyzb_{1}$, $0=xyzb_{1}\oplus a_{2}\oplus b_{2}\oplus c_{2}$ equating we get: $1\oplus a_{1}\oplus a_{2}\oplus b_{1}\oplus b_{2}\oplus1\oplus c_{1}\oplus c_{2}=xyzb_{1}$, $a\oplus b\oplus c=xyzb_{1}$, hence $\frac{1}{2}\left(P^{c}+P^{44}\right)$\\\\
$P^{c}P^{c}\rightarrow$ For box 1 we have $a_{1}\oplus b_{1}\oplus c_{1}=0$, for box 2 we have $a_{2}\oplus b_{2}\oplus c_{2}=0$, equating we get: $1\oplus a_{1}\oplus a_{2}\oplus b_{1}\oplus b_{2}\oplus1\oplus c_{1}\oplus c_{2}=0$, $a\oplus b\oplus c=0$, hence $P^{c}$\\\\
The final box is given by:
\begin{flalign*}
\scalemath{0.9}{P_{\delta}^{44}P_{\delta}^{44}}&\scalemath{0.9}{=P^{44}\left(\frac{3}{2}\delta-\delta^{2}\right)+P^{c}\left(1-\left(\frac{3}{2}\delta-\delta^{2}\right)\right)}
\end{flalign*}
Applying \eqref{commonineq} to the final box we get:
\begin{flalign*}
V^{\prime}=-4\delta^{2}+6\delta+7
%&&\qed
\end{flalign*}
\texttt{Class 45:}
\begin{flalign*}
\scalemath{0.9}{P_{\delta}^{45}P_{\delta}^{45}}&\scalemath{0.9}{=\delta^{2}P^{45}P^{45}+\delta(1-\delta)\left(P^{45}P^{c}+P^{c}P^{45}\right)+(1-\delta)^{2}P^{c}P^{c}}
\end{flalign*}
$P^{45}P^{45}\rightarrow$ For box 1 we have $a_{1}\oplus b_{1}\oplus c_{1}=x_{1}y_{1}\oplus x_{1}z_{1}=xy\oplus xz$, $0=xy\oplus xz\oplus a_{1}\oplus b_{1}\oplus c_{1}$, for box 2 we have $a_{2}\oplus b_{2}\oplus c_{2}=x_{2}y_{2}\oplus x_{2}z_{2}=x_{1}y_{1}b_{1}\oplus x_{1}z_{1}=xyb_{1}\oplus xz$,
$0=xyb_{1}\oplus xz\oplus a_{2}\oplus b_{2}\oplus c_{2}$, equating we get: $1\oplus a_{1}\oplus a_{2}\oplus b_{1}\oplus b_{2}\oplus1\oplus c_{1}\oplus c_{2}=xy\oplus xyb_{1}\oplus xz\oplus xz$, $a\oplus b\oplus c=xy\left(1\oplus b_{1}\right)$, hence $\frac{1}{2}\left(P^{xy}+P^{c}\right)$\\\\
$P^{45}P^{c}\rightarrow$ For box 1 we have $a_{1}\oplus b_{1}\oplus c_{1}=x_{1}y_{1}\oplus x_{1}z_{1}=xy\oplus xz$, $0=xy\oplus xz\oplus a_{1}\oplus b_{1}\oplus c_{1}$, for box 2 we have $a_{2}\oplus b_{2}\oplus c_{2}=0$, equating we get: $1\oplus a_{1}\oplus a_{2}\oplus b_{1}\oplus b_{2}\oplus1\oplus c_{1}\oplus c_{2}=xy\oplus xz$, $a\oplus b\oplus c=xy\oplus xz$, hence $P^{45}$\\\\
$P^{c}P^{45}\rightarrow$ For box 1 we have $a_{1}\oplus b_{1}\oplus c_{1}=0$, for box 2 we have $a_{2}\oplus b_{2}\oplus c_{2}=x_{2}y_{2}\oplus x_{2}z_{2}=x_{1}y_{1}b_{1}\oplus x_{1}z_{1}=xyb_{1}\oplus xz$, $0=xyb_{1}\oplus xz\oplus a_{2}\oplus b_{2}\oplus c_{2}$ equating we get: $1\oplus a_{1}\oplus a_{2}\oplus b_{1}\oplus b_{2}\oplus1\oplus c_{1}\oplus c_{2}=xyb_{1}\oplus xz$, $a\oplus b\oplus c=xyb_{1}\oplus xz$, hence $\frac{1}{2}\left(P^{xz}+P^{45}\right)$\\\\
$P^{c}P^{c}\rightarrow$ For box 1 we have $a_{1}\oplus b_{1}\oplus c_{1}=0$, for box 2 we have $a_{2}\oplus b_{2}\oplus c_{2}=0$, equating $a_{1}\oplus b_{1}\oplus c_{1}=a_{2}\oplus b_{2}\oplus c_{2}$, $1\oplus a_{1}\oplus a_{2}\oplus b_{1}\oplus b_{2}\oplus1\oplus c_{1}\oplus c_{2}=0$, $a\oplus b\oplus c=0$, hence $P^{c}$\\\\
The final box is given by:
\begin{flalign*}
\scalemath{0.9}{P_{\delta}^{45}P_{\delta}^{45}}&\scalemath{0.9}{=\delta\left(\frac{1}{2}\left(\delta P^{xy}+(1-\delta)P^{xz}\right)+\frac{1}{2}P_{1-\delta}^{45}\right)+(1-\delta)P_{\delta}^{45}}
\end{flalign*}
Applying \eqref{commonineq} to the final box we get:
\begin{flalign*}
V^{\prime}=-6\delta^{2}+9\delta+7
\end{flalign*}
\texttt{Class 46:}
\begin{flalign*}
\scalemath{0.9}{P_{\delta}^{46}P_{\delta}^{46}}&\scalemath{0.9}{=\delta^{2}P^{46}P^{46}+\delta(1-\delta)\left(P^{46}P^{c}+P^{c}P^{46}\right)+(1-\delta)^{2}P^{c}P^{c}}
\end{flalign*}
$P^{46}P^{46}\rightarrow$ For box 1 we have $a_{1}\oplus b_{1}\oplus c_{1}=x_{1}y_{1}\oplus y_{1}z_{1}\oplus x_{1}z_{1}=xy\oplus yz\oplus xz$, $0=xy\oplus yz\oplus xz\oplus a_{1}\oplus b_{1}\oplus c_{1}$, for box 2 we have $a_{2}\oplus b_{2}\oplus c_{2}=x_{2}y_{2}\oplus y_{2}z_{2}\oplus x_{2}z_{2}=x_{1}y_{1}b_{1}\oplus y_{1}z_{1}b_{1}\oplus x_{1}z_{1}=xyb_{1}\oplus yzb_{1}\oplus xz$, $0=xyb_{1}\oplus yzb_{1}\oplus xz\oplus a_{2}\oplus b_{2}\oplus c_{2}$,
equating we get: $1\oplus a_{1}\oplus a_{2}\oplus b_{1}\oplus b_{2}\oplus1\oplus c_{1}\oplus c_{2}=xyb_{1}\oplus yzb_{1}\oplus xz\oplus xy\oplus yz\oplus xz$, $a\oplus b\oplus c=xyb_{1}\oplus xy\oplus yzb_{1}\oplus yz$, $a\oplus b\oplus c=\left(1\oplus b_{1}\right)\left(xy\oplus yz\right)$,hence $\frac{1}{2}\left(P^{xy\oplus yz}+P^{c}\right)$\\\\
$P^{46}P^{c}\rightarrow$ For box 1 we have $a_{1}\oplus b_{1}\oplus c_{1}=x_{1}y_{1}\oplus y_{1}z_{1}\oplus x_{1}z_{1}=xy\oplus yz\oplus xz$,
$0=xy\oplus yz\oplus xz\oplus a_{1}\oplus b_{1}\oplus c_{1}$, for box 2 we have $a_{2}\oplus b_{2}\oplus c_{2}=0$, equating we get: $1\oplus a_{1}\oplus a_{2}\oplus b_{1}\oplus b_{2}\oplus1\oplus c_{1}\oplus c_{2}=xy\oplus yz\oplus xz$, $a\oplus b\oplus c=xy\oplus yz\oplus xz$, hence $P^{46}$\\\\
$P^{c}P^{46}\rightarrow$ For box 1 we have $a_{1}\oplus b_{1}\oplus c_{1}=0$, for box 2 we have $a_{2}\oplus b_{2}\oplus c_{2}=x_{2}y_{2}\oplus y_{2}z_{2}\oplus x_{2}z_{2}=x_{1}y_{1}b_{1}\oplus y_{1}z_{1}b_{1}\oplus x_{1}z_{1}=xyb_{1}\oplus yzb_{1}\oplus xz$, $0=xyb_{1}\oplus yzb_{1}\oplus xz\oplus a_{2}\oplus b_{2}\oplus c_{2}$, equating we get: $1\oplus a_{1}\oplus a_{2}\oplus b_{1}\oplus b_{2}\oplus1\oplus c_{1}\oplus c_{2}=xyb_{1}\oplus yzb_{1}\oplus xz$, $a\oplus b\oplus c=b_{1}\left(xy\oplus yz\right)\oplus xz$, hence $\frac{1}{2}\left(P^{xz}+P^{46}\right)$\\\\
$P^{c}P^{c}\rightarrow$ For box 1 we have $a_{1}\oplus b_{1}\oplus c_{1}=0$, for box 2 we have $a_{2}\oplus b_{2}\oplus c_{2}=0$, equating $a_{1}\oplus b_{1}\oplus c_{1}=a_{2}\oplus b_{2}\oplus c_{2}$, $1\oplus a_{1}\oplus a_{2}\oplus b_{1}\oplus b_{2}\oplus1\oplus c_{1}\oplus c_{2}=0$, $a\oplus b\oplus c=0$, hence $P^{c}$\\\\
The final box is given by:
\begin{flalign*}
\scalemath{0.9}{P_{\delta}^{46}P_{\delta}^{46}}&\scalemath{0.9}{=\delta\left(\frac{1}{2}\left(\delta P^{xy\oplus yz}+(1-\delta)P^{xz}\right)+\frac{1}{2}P_{1-\delta}^{46}\right)+(1-\delta)P_{\delta}^{46}}
\end{flalign*}
Applying \eqref{commonineq} to the final box we get:
\begin{flalign*}
&&V^{\prime}=-10\delta^{2}+9\delta+7 &&\qed
%&&\qed
\end{flalign*}
\end{proof}
For class 45 protocol 2 distills in the region of $0<\delta<\dfrac{5}{6}$ and both for class 44 and class 46 in the region of $0<\delta<\dfrac{1}{2}$, white curves in fig.~\ref{fig:figure11} refer to protocol 2.

\subsection{Unique protocols}
We present three different protocols that atleast distill for each of their respective classes:

\subsubsection{Protocol 3}
\begin{figure}[t]
  \centering
  \begin{tikzpicture}
    \draw (-4.2,-3.0) rectangle (4.2,3.0);
     \node[inner sep=0pt] (russell) at (0,0)
    {\includegraphics[width=0.45\textwidth]{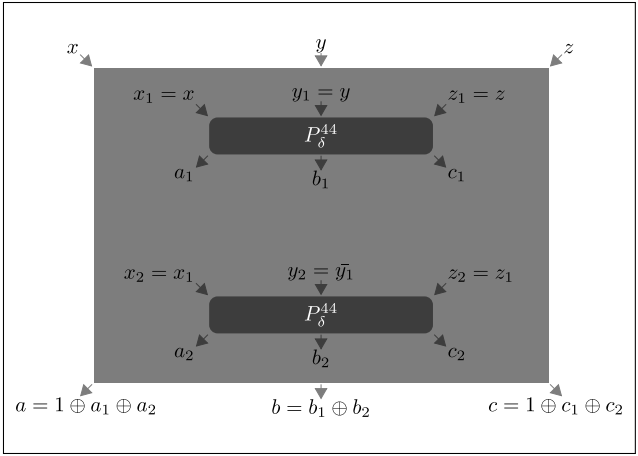}}; 
     \end{tikzpicture}
    \caption{Protocol 3}
  \label{fig:figure8}
\end{figure}
\begin{lemma}
Protocol 3(see fig.~\ref{fig:figure8}) distills for two copies of $P_{\delta}^{44}$
\end{lemma}
\begin{proof}
We start with two copies of $P_{\delta}^{44}$:
\begin{flalign*}
\scalemath{0.9}{P_{\delta}^{44}P_{\delta}^{44}}&\scalemath{0.9}{=\delta^{2}P^{44}P^{44}+\delta(1-\delta)\left(P^{44}P^{c}+P^{c}P^{44}\right)+(1-\delta)^{2}P^{c}P^{c}}
\end{flalign*}
$P^{44}P^{44}\rightarrow$ For box 1 we have $a_{1}\oplus b_{1}\oplus c_{1}=x_{1}y_{1}z_{1}=xyz$, $0=xyz\oplus a_{1}\oplus b_{1}\oplus c_{1}$, for box 2 we have $a_{2}\oplus b_{2}\oplus c_{2}=x_{2}y_{2}z_{2}=x_{1}(1\oplus y_{1})z_{1}=xz\oplus xyz$, $0=xz\oplus xyz\oplus a_{2}\oplus b_{2}\oplus c_{2}$ equating we get: $1\oplus a_{1}\oplus a_{2}\oplus b_{1}\oplus b_{2}\oplus1\oplus c_{1}\oplus c_{2}=xyz\oplus xz\oplus xyz$, $a\oplus b\oplus c=xz$, hence $P^{xz}$\\\\
$P^{44}P^{c}\rightarrow$ For box 1 we have $a_{1}\oplus b_{1}\oplus c_{1}=x_{1}y_{1}z_{1}=xyz$, $0=xyz\oplus a_{1}\oplus b_{1}\oplus c_{1}$, for box 2 we have $a_{2}\oplus b_{2}\oplus c_{2}=0$, equating we get: $1\oplus a_{1}\oplus a_{2}\oplus b_{1}\oplus b_{2}\oplus1\oplus c_{1}\oplus c_{2}=xyz$, $a\oplus b\oplus c=xyz$, hence $P^{44}$\\\\
$P^{c}P^{44}\rightarrow$ For box 1 we have $a_{1}\oplus b_{1}\oplus c_{1}=0$, for box 2 we have $a_{2}\oplus b_{2}\oplus c_{2}=x_{2}y_{2}z_{2}=x_{1}(1\oplus y_{1})z_{1}=xz\oplus xyz$, $0=xz\oplus xyz\oplus a_{2}\oplus b_{2}\oplus c_{2}$ equating we get: $1\oplus a_{1}\oplus a_{2}\oplus b_{1}\oplus b_{2}\oplus1\oplus c_{1}\oplus c_{2}=xz\oplus xyz$, $a\oplus b\oplus c=xz\oplus xyz$, hence $P^{xz\oplus xyz}$\\\\
$P^{c}P^{c}\rightarrow$ For box 1 we have $a_{1}\oplus b_{1}\oplus c_{1}=0$, for box 2 we have $a_{2}\oplus b_{2}\oplus c_{2}=0$, equating we get: $1\oplus a_{1}\oplus a_{2}\oplus b_{1}\oplus b_{2}\oplus1\oplus c_{1}\oplus c_{2}=0$, $a\oplus b\oplus c=0$, hence $P^{c}$\\\\
The final box is given by:
\begin{flalign*}
P_{\delta}^{44}P_{\delta}^{44}&=\delta\left(\delta P^{xz}+(1-\delta)P^{xz\oplus xyz}\right)+(1-\delta)P_{\delta}^{44}
\end{flalign*}
Applying \eqref{commonineq} to the final box we get:
\begin{flalign*}
&&V^{\prime}=6\delta+7 &&\qed
\end{flalign*}
\end{proof}
The protocol distills in the region of $0<\delta\leqslant1$.

\subsubsection{Protocol 4}
\begin{figure}[t]
  \centering
  \begin{tikzpicture}
    \draw (-4.2,-3.0) rectangle (4.2,3.0);
        \node[inner sep=0pt] (russell) at (0,0)
    {\includegraphics[width=0.45\textwidth]{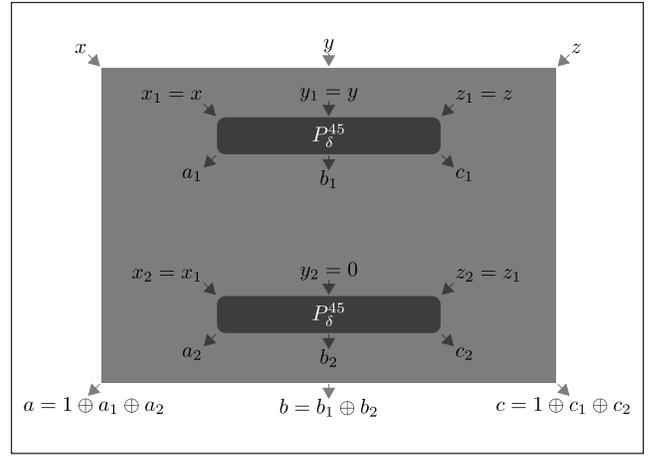}}; 
      \end{tikzpicture}
    \caption{Protocol 4}
  \label{fig:figure9}
\end{figure}
\begin{lemma}
Protocol 4(see fig.~\ref{fig:figure9}) distills for two copies of $P_{\delta}^{45}$
\end{lemma}
\begin{proof}
We start with two copies of $P_{\delta}^{45}$:
\begin{flalign*}
\scalemath{0.9}{P_{\delta}^{45}P_{\delta}^{45}}&\scalemath{0.9}{=\delta^{2}P^{45}P^{45}+\delta(1-\delta)\left(P^{45}P^{c}+P^{c}P^{45}\right)+(1-\delta)^{2}P^{c}P^{c}}
\end{flalign*}
$P^{45}P^{45}\rightarrow$ For box 1 we have $a_{1}\oplus b_{1}\oplus c_{1}=x_{1}y_{1}\oplus x_{1}z_{1}=xy\oplus xz$, $0=xy\oplus xz\oplus a_{1}\oplus b_{1}\oplus c_{1}$, for box 2 we have $a_{2}\oplus b_{2}\oplus c_{2}=x_{2}y_{2}\oplus x_{2}z_{2}=0\oplus x_{1}z_{1}=xz$, $0=xz\oplus a_{2}\oplus b_{2}\oplus c_{2}$, equating we get: $1\oplus a_{1}\oplus a_{2}\oplus b_{1}\oplus b_{2}\oplus1\oplus c_{1}\oplus c_{2}=xy\oplus xz\oplus xz$, $a\oplus b\oplus c=xy$, hence $P^{xy}$\\\\
$P^{45}P^{c}\rightarrow$ For box 1 we have $a_{1}\oplus b_{1}\oplus c_{1}=x_{1}y_{1}\oplus x_{1}z_{1}=xy\oplus xz$, $0=xy\oplus xz\oplus a_{1}\oplus b_{1}\oplus c_{1}$, for box 2 we have $a_{2}\oplus b_{2}\oplus c_{2}=0$, equating we get: $1\oplus a_{1}\oplus a_{2}\oplus b_{1}\oplus b_{2}\oplus1\oplus c_{1}\oplus c_{2}=xy\oplus xz$, $a\oplus b\oplus c=xy\oplus xz$, hence $P^{45}$\\\\
$P^{c}P^{45}\rightarrow$ For box 1 we have $a_{1}\oplus b_{1}\oplus c_{1}=0$, for box 2 we have $a_{2}\oplus b_{2}\oplus c_{2}=x_{2}y_{2}\oplus x_{2}z_{2}=0\oplus x_{1}z_{1}=xz$, $0=xz\oplus a_{2}\oplus b_{2}\oplus c_{2}$ equating we get: $1\oplus a_{1}\oplus a_{2}\oplus b_{1}\oplus b_{2}\oplus1\oplus c_{1}\oplus c_{2}=xz$, $a\oplus b\oplus c=xz$, hence $P^{xz}$\\\\
$P^{c}P^{c}\rightarrow$ for box 1 we have $a_{1}\oplus b_{1}\oplus c_{1}=0$, for box 2 we have $a_{2}\oplus b_{2}\oplus c_{2}=0$, equating $a_{1}\oplus b_{1}\oplus c_{1}=a_{2}\oplus b_{2}\oplus c_{2}$, $1\oplus a_{1}\oplus a_{2}\oplus b_{1}\oplus b_{2}\oplus1\oplus c_{1}\oplus c_{2}=0$, $a\oplus b\oplus c=0$, hence $P^{c}$\\\\
The final box is given by:
\begin{flalign*}
P_{\delta}^{45}P_{\delta}^{45}&=\delta\left(\delta P^{xy}+(1-\delta)P^{xz}\right)+(1-\delta)P_{\delta}^{45}
\end{flalign*}
Applying \eqref{commonineq} to the final box we get:
\begin{flalign*}
&&V^{\prime}=-4\delta^{2}+10\delta+7 &&\qed
%&&\qed
\end{flalign*}
\end{proof}
The region of distillation for protocol 4 is $0<\delta\leqslant1$.

\subsubsection{Protocol 5}
\begin{figure}[t]
  \centering
  \begin{tikzpicture}
    \draw (-4.2,-3.0) rectangle (4.2,3.0);
       \node[inner sep=0pt] (russell) at (0,0)
    {\includegraphics[width=0.45\textwidth]{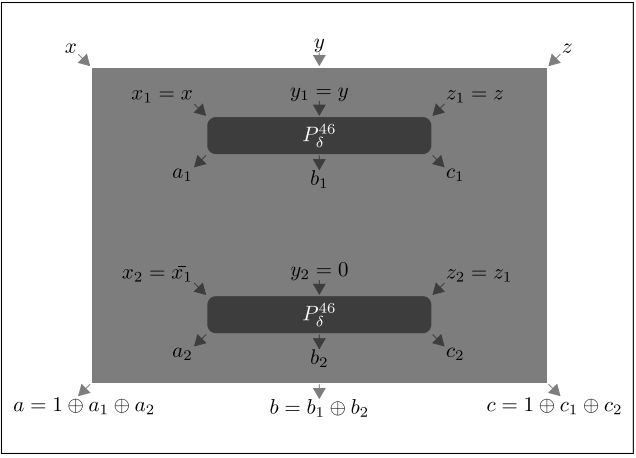}}; 
     \end{tikzpicture}
    \caption{Protocol 5}
  \label{fig:figure10}
\end{figure}
\begin{lemma}
Protocol 5(see fig.~\ref{fig:figure10}) distills for two copies of $P_{\delta}^{46}$
\end{lemma}
\begin{proof}
We start with two copies of $P_{\delta}^{46}$:
\begin{flalign*}
\scalemath{0.9}{P_{\delta}^{46}P_{\delta}^{46}}&\scalemath{0.9}{=\delta^{2}P^{46}P^{46}+\delta(1-\delta)\left(P^{46}P^{c}+P^{c}P^{46}\right)+(1-\delta)^{2}P^{c}P^{c}}
\end{flalign*}
$P^{46}P^{46}\rightarrow$ For box 1 we have $a_{1}\oplus b_{1}\oplus c_{1}=x_{1}y_{1}\oplus y_{1}z_{1}\oplus x_{1}z_{1}=xy\oplus yz\oplus xz$,
$0=xy\oplus yz\oplus xz\oplus a_{1}\oplus b_{1}\oplus c_{1}$, for box 2 we have $a_{2}\oplus b_{2}\oplus c_{2}=x_{2}y_{2}\oplus y_{2}z_{2}\oplus x_{2}z_{2}=0\oplus0\oplus(1\oplus x_{1})z_{1}=(1\oplus x)z=z\oplus xz$, $0=z\oplus xz\oplus a_{2}\oplus b_{2}\oplus c_{2}$, equating we get: $1\oplus a_{1}\oplus a_{2}\oplus b_{1}\oplus b_{2}\oplus1\oplus c_{1}\oplus c_{2}=xy\oplus yz\oplus xz\oplus z\oplus xz$, $a\oplus b\oplus c=xy\oplus yz\oplus z$, $a\oplus b\oplus c=xy\oplus yz\oplus z$, hence $P^{xy\oplus yz\oplus z}$\\\\
$P^{46}P^{c}\rightarrow$ For box 1 we have $a_{1}\oplus b_{1}\oplus c_{1}=x_{1}y_{1}\oplus y_{1}z_{1}\oplus x_{1}z_{1}=xy\oplus yz\oplus xz$,$0=xy\oplus yz\oplus xz\oplus a_{1}\oplus b_{1}\oplus c_{1}$, for box 2 we have $a_{2}\oplus b_{2}\oplus c_{2}=0$, equating we get: $1\oplus a_{1}\oplus a_{2}\oplus b_{1}\oplus b_{2}\oplus1\oplus c_{1}\oplus c_{2}=xy\oplus yz\oplus xz$, $a\oplus b\oplus c=xy\oplus yz\oplus xz$, hence $P^{46}$\\\\
$P^{c}P^{46}\rightarrow$ For box 1 we have $a_{1}\oplus b_{1}\oplus c_{1}=0$, for box 2 we have $a_{2}\oplus b_{2}\oplus c_{2}=x_{2}y_{2}\oplus y_{2}z_{2}\oplus x_{2}z_{2}=0\oplus0\oplus(1\oplus x_{1})z_{1}=(1\oplus x)z=z\oplus xz$, $0=z\oplus xz\oplus a_{2}\oplus b_{2}\oplus c_{2}$ equating we get:  $1\oplus a_{1}\oplus a_{2}\oplus b_{1}\oplus b_{2}\oplus1\oplus c_{1}\oplus c_{2}=z\oplus xz$, $a\oplus b\oplus c=z\oplus xz$, hence $P^{z\oplus xz}$\\\\
$P^{c}P^{c}\rightarrow$ for box 1 we have $a_{1}\oplus b_{1}\oplus c_{1}=0$, for box 2 we have $a_{2}\oplus b_{2}\oplus c_{2}=0$, equating we get:  $1\oplus a_{1}\oplus a_{2}\oplus b_{1}\oplus b_{2}\oplus1\oplus c_{1}\oplus c_{2}=0$, $a\oplus b\oplus c=0$, hence $P^{c}$\\\\
The final box is given by:
\begin{flalign*}
P_{\delta}^{46}P_{\delta}^{46}&=\delta\left(\delta P^{xy\oplus yz\oplus z}+(1-\delta)P^{z\oplus xz}\right)+(1-\delta)P_{\delta}^{46}
\end{flalign*}
Applying \eqref{commonineq} to the final box we get:
\begin{flalign*}
&&V^{\prime}=8\delta^{2}-2\delta+7 &&\qed
%&&\qed
\end{flalign*}
\end{proof}
The region of distillation for protocol 5 is $\dfrac{3}{4}<\delta\leqslant1$.\\\\

From fig.~\ref{fig:figure11} we can conclude that protocol 2 as compared to protocol 1 performs much better for class 45, for class 46 protocol 2 performs slightly better than protocol 1, however for class 44 protocol 1 performs better than protocol 2, each of the unique protocols represented by black curves in fig.~\ref{fig:figure11} perform better than common protocols for their respective classes and distillation regions, among unique protocols protocol 4 for class 45 performs the best.

\begin{widetext}
\begin{figure}[t]
  \centering
  \begin{tikzpicture}
    \draw (-8.5,-5.5) rectangle (8.5,5.5);
     \node[inner sep=0pt] (russell) at (0,0)
    {\includegraphics[width=0.9\textwidth]{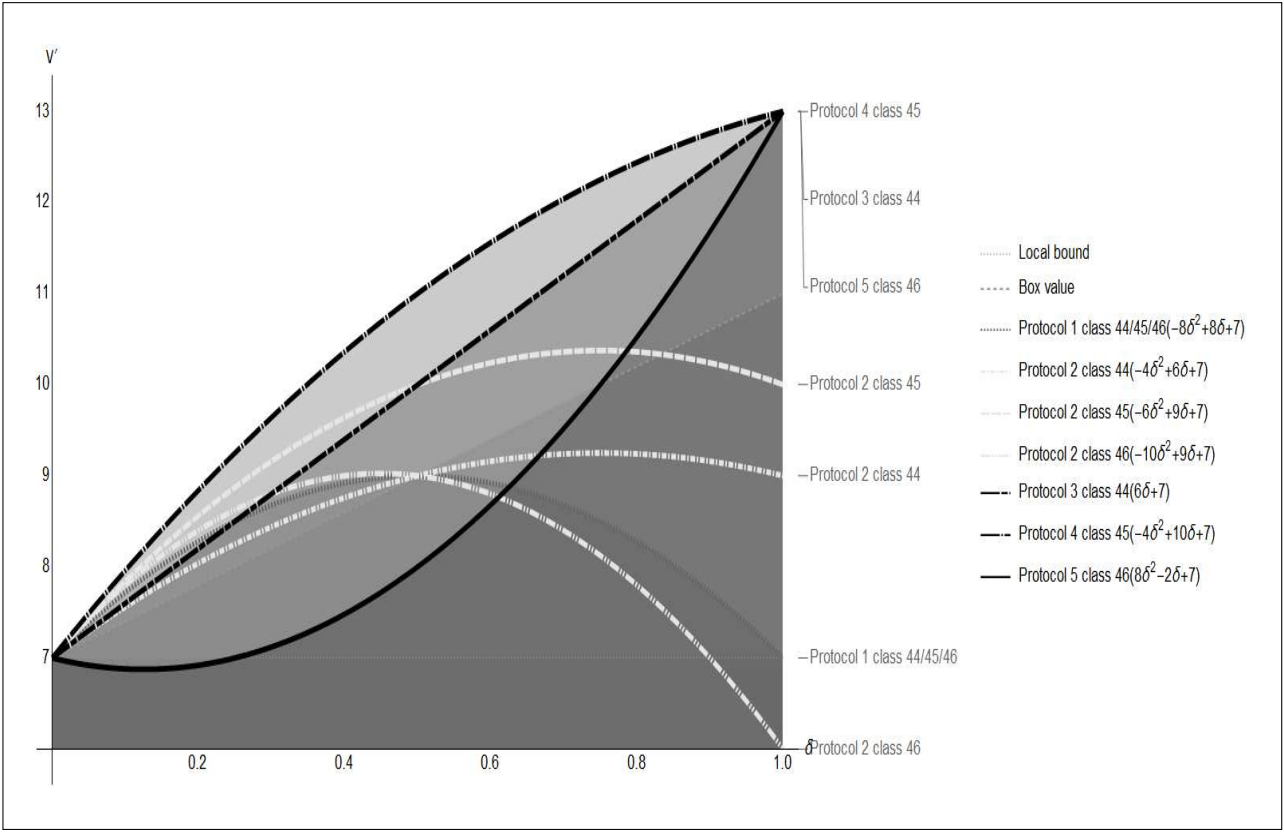}}; 
      \end{tikzpicture}
   \caption{Plot for $P_{\delta}^{N}$ distillation protocols}
  \label{fig:figure11}
\end{figure}

\section{Final comments:\protect}
We presented here a limited study of distillation in tripartite NLBs, in comparison to the bipartite($(2,2,2)$) case the tripartite($(3,2,2)$) case is much more complicated with 53856 extremal points classified into 46 classes\cite{1751-8121-44-6-065303},  because of this  complexity we selected a limited part of tripartite no-signalling polytope specifically a representative of each class 44,45,46 and GHZ box, the reason for this selection is based on a common criteria that each of these genuine tripartite boxes in their definition have no one and two party expectations. As is the case with tripartite extremal points the distillation is also complicated and does lead to numerous distillation protocols, for GHZ box we selected a pair of class 2 inequalities while for $P_{\delta}^{N}$ we selected a class 41 inequality out of 53856 inequalities\cite{quant-ph/0305190}, this inequality has a limited no-signalling violation for $P^{N}$ which provides a common box value for $P^{N}$, while GHZ box is distilled with only single protocol, $P_{\delta}^{N}$ has much more distillation protocols, for $P_{\delta}^{N}$ we selected two common protocols, protocol 2 is selected for large no-signalling violation of $P_{\delta}^{45}$, with the same selection criteria we selected unique protocols, this has led to the conjecture that common protocols cannot outperform unique protocols, clearly in our setting as far as no-signalling violation is  concerned class 45 has outperformed other classes. In future we intend to study distillation in other classes however such a task will require a suitable representation of boxes in terms of boolean polynomials which while is not impossible but still is much more complicated, of course such a exercise will lead to numerous protocols, which perhaps can be classified according to some criteria.

\begin{acknowledgments}
\begin{itemize}
\item I am thankful to Jibran Rashid for his initial guidance.
\item I thank Cezary $\acute{\textnormal{S}}$liwa for his discussion on bell inequalities.
\end{itemize}
\end{acknowledgments}

\end{widetext}

\bibliography{tripartite_pre}

\end{document}